\documentclass{article}
\usepackage{amsmath,amssymb}
\def \E{\mathbb E}
\def \P{\mathbb P}
\def \R{\mathbb R}
\def \S{{\cal S}}

\def \1{{\bf 1}}
\def \bX{{\bf X}}

\def \p0{{\frac{128\, \beta \max\{\mu_1^2,\mu_0\}}{\nu_0} \ \frac{r\, \log^2( n)}{n}}}

\usepackage{graphicx}
\usepackage{url,epsfig}
\usepackage{algorithm,algorithmic}
\usepackage{colortbl}
\usepackage{amsmath, amssymb, amsthm}
\usepackage{times}

\def\R{{\mathbb R}}
\newtheorem{theorem}{Theorem}[section]

\newtheorem{lemma}{Lemma}

\newtheorem{defNew}{Definition}

\title{High-Rank Matrix Completion and \\ Subspace Clustering with Missing Data}

\author{
Brian Eriksson\thanks{The first two authors contributed equally to this paper.} \\ Boston University and \\ University of Wisconsin - Madison \\ eriksson@cs.bu.edu
\and
Laura Balzano{$^*$}\\ University of Wisconsin - Madison \\ sunbeam@ece.wisc.edu
\and
Robert Nowak \\ University of Wisconsin - Madison \\ nowak@ece.wisc.edu
}
\begin{document}

%

%
\date{December 2011}

\maketitle


\begin{abstract}
This paper considers the problem of completing a matrix with many missing entries under the assumption that the columns of the matrix belong to a union of multiple low-rank subspaces. This generalizes the standard {\em low-rank matrix completion} problem to situations in which the matrix rank can be quite high or even full rank.  Since the columns belong to a union of subspaces, this problem may also be viewed as a missing-data version of the {\em subspace clustering} problem.  Let $\bX$ be an $n\times N$ matrix whose (complete) columns lie in a union of at most $k$ subspaces, each of rank $\leq r < n$, and assume $N\gg kn$.
The main result of the paper shows that under mild assumptions each column of $\bX$ can be perfectly recovered with high probability from an incomplete version so long as at least $C r N \log^2(n)$ entries of $\bX$ are observed uniformly at random, with $C>1$ a constant depending on the usual incoherence conditions, the geometrical arrangement of subspaces, and the distribution of columns over the subspaces.
The result is illustrated with numerical experiments and an application to Internet distance matrix completion and topology identification.
\end{abstract}

\section{Introduction}

Consider a real-valued $n \times N$ dimensional matrix $\bX$.
Assume that the columns of $\bX$ lie in the union of at most $k$ subspaces of $\R^{n}$, each having dimension at most $r < n$ and assume that $N>kn$. We are especially interested in  ``high-rank'' situations in which the total rank (the rank of the union of the subspaces) may be $n$.  Our goal is to complete $\bX$ based on an observation of a small random subset of its entries. We propose a novel method for this matrix completion problem.  In the applications we have in mind $N$ may be arbitrarily large, and so we will focus on quantifying the probability that a given column is perfectly completed, rather than the probability that whole matrix is perfectly completed ({\em i.e.,} every column is perfectly completed).  Of course it is possible to translate between these two quantifications using a union bound, but that bound becomes meaningless if $N$ is extremely large.

Suppose the entries of $\bX$ are observed uniformly at random with probability $p_0$.
Let $\Omega$ denote the set of indices of observed entries and let $\bX_{\Omega}$ denote the observations of $\bX$.  Our main result shows that under a mild set of assumptions each column of $\bX$ can be perfectly recovered from $\bX_{\Omega}$ with high probability using a computationally efficient procedure if
\begin{eqnarray}
p_0 \geq C \, \frac{r}{n}  \log^2(n)
\label{eqn:mainBound1}
\end{eqnarray}
where $C>1$ is a constant depending on the usual incoherence conditions as well as the geometrical arrangement of subspaces and the distribution of the columns in the subspaces.

\subsection{Connections to Low-Rank Completion}
Low-rank matrix completion theory \cite{mcRecht} shows that an \mbox{$n\times N$} matrix of rank $r$ can be recovered from incomplete observations, as long as the number of entries observed (with locations sampled uniformly at random) exceeds $rN\log^2N$ (within a constant factor and assuming $n \leq N$).  It is also known that, in the same setting, completion is impossible if the number of observed entries is less than a constant times $rN \log N$ \cite{candes-tao}.  These results imply that if the rank of $\bX$ is close to $n$, then all of the entries are needed in order to determine the matrix.

Here we consider a matrix whose columns lie in the union of at most $k$ subspaces of $\R^{n}$.  Restricting the rank of each subspace to at most $r$, then the rank of the full matrix our situation could be as large as $kr$, yielding the requirement $krN \log^2 N$ using current matrix completion theory. In contrast, the bound in (\ref{eqn:mainBound1}) implies that the completion of each column is possible from a constant times $ r N \log^2 n$ entries sampled uniformly at random. Exact completion of every column can be guaranteed by replacing $\log^2n$ with $\log^2N$ is this bound, but since we allow $N$ to be very large we prefer to state our result in terms of per-column completion. Our method, therefore, improves significantly upon conventional low-rank matrix completion, especially when $k$ is large. This does not contradict the lower bound in \cite{candes-tao}, because the matrices we consider are not arbitrary high-rank matrices, rather the columns must belong to a union of rank $\leq r$ subspaces.

\subsection{Connections to Subspace Clustering}
Let $x_1,\dots,x_{N} \in \R^{n}$ and assume each $x_i$ lies in one of at most $k$ subspaces of $\R^{n}$.  Subspace clustering is the problem of learning the subspaces from $\{x_i\}_{i=1}^{N}$ and assigning each vector to its proper subspace; cf. \cite{vidal} for a overview.  This is a challenging problem, both in terms of computation and inference, but provably probably correct subspace clustering algorithms now exist \cite{kanatani01, gpcaJournal05, lerman11}.  Here we consider the problem of \emph{high rank matrix completion}, which is essentially equivalent to subspace clustering with missing data. This problem has been looked at in previous works \cite{weiss04, gpcaMissingData08}, but to the best of our knowledge our method and theoretical bounds are novel. Note that our sampling probability bound (\ref{eqn:mainBound1}) requires that only slightly more than $r$ out of $n$ entries are observed in each column, so the matrix may be highly incomplete.


\subsection{A Motivating Application}
There are many applications of subspace clustering, and it is reasonable to suppose that data may often be missing in high-dimensional problems.  One such application is the Internet distance matrix completion and topology identification problem.  Distances between networked devices can be measured in terms of hop-counts, the number of routers between the devices.  Infrastructures exist that record distances from
$N$ end host computers to a set of $n$ monitoring points throughout the Internet.  The complete set of distances determines the network topology between the computers and the monitoring points \cite{EBN08}. These infrastructures are based entirely on passively monitoring of normal traffic.  One advantage is the ability to monitor a very large portion of the Internet, which is not possible using active probing methods due to the burden they place on networks.  The disadvantage of passive monitoring is that measurements collected are based on normal traffic, which is not specifically designed or controlled, therefore a subset of the distances may not be observed.  This poses a matrix completion problem, with the incomplete distance matrix being potentially full-rank in this application. However, computers tend to be clustered within subnets having a small number of egress (or access) points to the Internet at large.  The number of egress points in a subnet limits the rank of the submatrix of distances from computers in the subnet to the monitors.  Therefore the columns of the $n\times N$ distance matrix lie in the union of $k$ low-rank subspaces, where $k$ is the number of subnets.  The solution to the matrix completion problem yields all the distances (and hence the topology) as well as the subnet clusters.

\subsection{Related Work}
\label{sec:relWork}

The proof of the main result draws on ideas from matrix completion theory, subspace learning and detection with missing data, and subspace clustering.
One key ingredient in our approach is the celebrated results on low-rank Matrix Completion~\cite{mcRecht,candes-tao, candes-recht}.  Unfortunately, in many real-world problems where missing data is present, particularly when the data is generated from a union of subspaces, these matrices can have very large rank values ({\em e.g.,} networking data in \cite{domainImpute}).  Thus, these prior results will require effectively all the elements be observed to accurately reconstruct the matrix.

Our work builds upon the results of \cite{balzanoSubspace}, which quantifies the deviation of an incomplete vector norm with respect to the incoherence of the sampling pattern.  While this work also examines subspace detection using incomplete data, it assumes complete knowledge of the subspaces.

While research that examines subspace learning has been presented in \cite{mauroSubspace}, the work in this paper differs by the concentration on learning from incomplete observations ({\em i.e.,} when there are missing elements in the matrix), and by the methodological focus ({\em i.e.,} nearest neighbor clustering versus a multiscale Singular Value Decomposition approach).

\subsection{Sketch of Methodology}
The algorithm proposed in this paper involves several relatively intuitive steps, outlined below. We go into detail for each of these steps in following sections.
\\ \vspace{-.1in} \\
\noindent{\bf Local Neighborhoods.} A subset of columns of $\bX_{\Omega}$ are selected uniformly at random. These are called {\em seeds}.  A set of nearest neighbors is identified for each seed from the remainder of $\bX_{\Omega}$.  In Section~\ref{sec:localneighborhoods}, we show that nearest neighbors can be reliably identified, even though a large portion of the data are missing, under the usual incoherence assumptions. \\ \vspace{-.1in} \\
\noindent{\bf Local Subspaces.} The subspace spanned by each seed and its neighborhood is identified using matrix completion.  If matrix completion fails ({\em i.e.,} if the resulting matrix does not agree with the observed entries and/or the rank of the result is greater than $r$), then the seed and its neighborhood are discarded.  In Section~\ref{sec:localsubspaces} we show that when the number of seeds and the neighborhood sizes are large enough, then with high probability all $k$ subspaces are identified. We may also identify additional subspaces which are unions of the true subspaces, which leads us to the next step. An example of these neighborhoods is shown in Figure~\ref{fig:nnExample}.\\ \vspace{-.1in} \\
\noindent{\bf Subspace Refinement.} The set of subspaces obtained from the matrix completions is pruned to remove all but $k$ subspaces. The pruning is accomplished by simply discarding all subspaces that are spanned by the union of two or more other subspaces. This can be done efficiently, as is shown in Section~\ref{sec:subspacerefine}. \\ \vspace{-.1in} \\
\noindent{\bf Full Matrix Completion.} Each column in $\bX_{\Omega}$ is assigned to its proper subspace and completed by projection onto that subspace, as described in Section~\ref{sec:fullmonty}.  Even when many observations are missing, it is possible to find the correct subspace and the projection using results from subspace detection with missing data~\cite{balzanoSubspace}.  The result of this step is a completed matrix $\widehat{\bX}$ such that each column is correctly completed with high probability. \\ \vspace{-.1in} \\
The mathematical analysis will be presented in the next few sections, organized according to these steps.  After proving the main result,  experimental results are presented in the final section.

\begin{figure}[htb]
\begin{minipage}[h]{0.99\linewidth}
\centerline{\includegraphics[width=8cm]{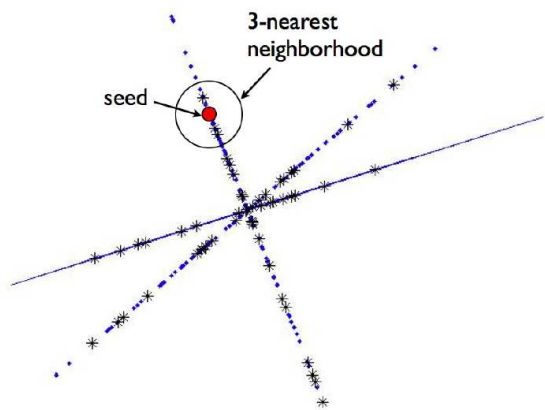}}
\end{minipage}
{
\caption{\label{fig:nnExample} \small Example of nearest-neighborhood selecting points on from a single subspace. For illustration, samples from three one-dimensional subspaces are depicted as small dots.  The large dot is the seed.  The subset of samples with significant observed support in common with that of the seed are depicted by $\ast$'s. If the density of points is high enough, then the nearest neighbors we identify will belong to the same subspace as the seed.  In this case we depict the ball containing the $3$ nearest neighbors of the seed with significant support overlap. }}
\label{ss1}
\end{figure}


\section{Key Assumptions and Main Result}

The notion of incoherence plays a key role in matrix completion and subspace recovery from incomplete observations.

\begin{defNew}
The \emph{coherence} of an $r$-dimensional subspace ${\cal S} \subseteq \R^n$ is $$\mu({\cal S}) := \frac{n}{r} \max_j \|P_{\cal S} e_j\|_2^2$$
where $P_{\cal S}$ is the projection operator onto ${\cal S}$ and $\{e_j\}$ are the canonical unit vectors for $\R^n$.
\end{defNew}
Note that $1 \leq \mu(\S) \leq n/r$.
The coherence of single vector $x\in \R^n$ is $\mu(x)=\frac{n \|x\|_\infty^2}{\|x\|_2^2} $,
which is precisely the coherence of the one-dimensional subspace spanned by $x$.
With this definition, we can state the main assumptions we make about the matrix $\bX$.
\begin{description}
\item[A1.] The columns of $\bX$ lie in the union of at most $k$ subspaces, with $k=o(n^d)$ for some $d>0$. The subspaces are denoted by $\S_1,\dots,\S_k$ and each has rank at most $r<n$.  The $\ell_2$-norm of each column is $\leq 1$.
\item[A2.] The coherence of each subspace is bounded above by $\mu_0$.  The coherence of each column is bounded above by $\mu_1$ and for any pair of columns, $x_1$ and $x_2$, the coherence of $x_1-x_2$ is also bounded above by $\mu_1$.
\item[A3.] The columns of $\bX$ do not lie in the intersection(s) of the subspaces with probability $1$, and if $\mbox{rank}(\S_i)=r_i$, then any subset of $r_i$ columns from $\S_i$ spans $\S_i$ with probability $1$. Let $0<\epsilon_0<1$ and $\S_{i,\epsilon_0}$ denote the subset of points in $\S_i$ at least $\epsilon_0$ distance away from any other subspace.  There exists a constant $0 < \nu_0 \leq 1$,
depending on $\epsilon_0$, such that \vspace{-.05in}
\begin{description}
\item[(i)] The probability that a column selected uniformly at random belongs to $\S_{i,\epsilon_0}$  is at least $\nu_0/k$.
\item[(ii)]If $x \in \S_{i,\epsilon_0}$, then the probability that a column selected uniformly at random belongs to the ball of radius $\epsilon_0$ centered at $x$ is at least $\nu_0 \epsilon_0^r/k$.
\end{description}
\end{description} \vspace{-.075in}

The conditions of {\bf A3} are met if, for example, the columns are drawn from a mixture of continuous distributions on each of the subspaces. The value of $\nu_0$ depends on the geometrical arrangement of the subspaces and the distribution of the columns within the subspaces.  If the subspaces are not too close to each other, and the distributions within the subspaces are fairly uniform, then typically $\nu_0$ will be not too close to $0$.    We define three key quantities, the confidence parameter $\delta_0$, the required number of ``seed'' columns $s_0$, and a quantity $\ell_0$ related to the neighborhood formation process (see Algorithm~\ref{alg:localNeighbor} in Section~\ref{sec:localneighborhoods}):
\begin{eqnarray}
\delta_0 & := & n^{2-2\beta^{1/2}}\log n \ , \ \mbox{for some $\beta>1$} \ , \label{delta} \\
s_0 & := & \left\lceil \frac{k(\log k + \log 1/\delta_0)}{(1-e^{-4})\nu_0} \right\rceil  \ , \nonumber \\
\ell_0 & := & \left\lceil \max\left\{\frac{2k}{\nu_0 (\frac{\epsilon_0}{\sqrt{3}})^r} \, , \, \frac{8k  \log(s_0/\delta_0)}{ n \nu_0 (\frac{\epsilon_0}{\sqrt{3}})^r} \right\} \right\rceil  \ . \nonumber
\end{eqnarray}
We can now state the main result of the paper.
\begin{theorem}  Let $\bX$ be an $n\times N$ matrix satisfying {\bf{A1}}-{\bf{A3}}.
  Suppose that each entry of $\bX$ is observed independently with probability $p_0$.  If
\begin{eqnarray*}
p_0 \ \geq  \ \p0
\label{bound}
\end{eqnarray*}
and
\begin{eqnarray*}
N & \geq & \ell_0 n(2\delta_0^{-1}s_0\ell_0 n)^{\mu_0^2\log p_0^{-1}}
\end{eqnarray*}
then each column of $\bX$ can be perfectly recovered with probability at least $1- (6+15s_0) \, \delta_0$, using the methodology sketched above (and detailed later in the paper).
\label{mainthm}
\end{theorem}
The requirements on sampling are essentially the same as those for standard low-rank matrix completion, apart from requirement that the total number of columns $N$ is sufficiently large.
This is needed to ensure that each of the subspaces is sufficiently represented in the matrix.  The requirement on $N$ is polynomial in $n$ for fixed $p_0$, which is easy to see based on the definitions of $\delta_0$, $s_0$, and $\ell_0$ (see further discussion at the end of Section~\ref{sec:localneighborhoods}).

Perfect recovery of each column is guaranteed with probability that decreases linearly  in $s_0$,
which itself is linear in $k$ (ignoring log factors). This is expected since this problem is more difficult than $k$ individual low-rank matrix completions.  We state our results in terms of a per-column (rather than full matrix) recovery guarantee. A full matrix recovery guarantee can be given by replacing $\log^2n$ with $\log^2 N$.  This is evident from the final completion step discussed in Lemma~\ref{thm:projection}, below.  However, since $N$ may be quite large (perhaps arbitrarily large) in the applications we envision, we chose to state our results in terms of a per-column guarantee.

The details of the methodology and lemmas leading to the theorem above are developed in the subsequent sections following the four steps of the methodology outlined above.  In certain cases it will be more convenient to consider sampling the locations of observed entries uniformly at random {\em with replacement} rather than without replacement, as assumed above.   The following lemma will be useful for translating bounds derived assuming sampling with replacement to our situation
(the same sort of relation is noted in Proposition~3.1~in~\cite{mcRecht}).
\begin{lemma}
Draw $m$ samples independently and uniformly from $\{1,\dots,n\}$ and let $\Omega'$ denote the resulting subset of unique values.  Let $\Omega_m$ be a subset of size $m$ selected uniformly at random from $\{1,\dots,n\}$. Let $E$ denote an event depending on a random subset of $\{1,\dots,n\}$.  If $\P(E(\Omega_m))$ is a non-increasing function of $m$, then $\P(E(\Omega')) \geq \P(E(\Omega_m))$.
\label{meta}
\end{lemma} 
\begin{proof}
For $k=1,\dots,m$, let $\Omega_k$ denote a subset of size $k$ sampled uniformly at random from $\{1,\dots,n\}$,
and let $m'=|\Omega'|$.
\begin{eqnarray*}
\P(E(\Omega')) &= & \sum_{k=0}^m \P\left(E(\Omega')\, | \, m'=k\right) \P(m'=k) \\
& = & \sum_{k=0}^m \P(E(\Omega_k)) \P(m'=k) \\
& \geq & \P(E(\Omega_m)) \sum_{k=0}^m  \P(m'=k) \ .
\end{eqnarray*}
\end{proof}

\section{Local Neighborhoods}
\label{sec:localneighborhoods}

In this first step, $s$ columns of ${\bf X}_\Omega$ are selected uniformly at random and a set of ``nearby'' columns are identified for each, constituting a local neighborhood of size $n$. All bounds that hold are designed with probability at least $1-\delta_0$, where $\delta_0$ is defined in (\ref{delta}) above.   The $s$ columns are called ``seeds.'' The required size of $s$ is determined as follows.
\begin{lemma}
Assume {\bf{A3}} holds. If the number of chosen seeds,
\begin{eqnarray*}
s & \geq & \frac{k(\log k + \log 1/\delta_0)}{(1-e^{-4})\nu_0} \ ,
\end{eqnarray*}
then with probability greater than $1-\delta_0$ for each $i=1,\dots,k$, at least one seed is in $\S_{i,\epsilon_0}$ and each seed column has at least
\begin{eqnarray}
\eta_0 := {\frac{64\, \beta \max\{\mu_1^2,\mu_0\}}{\nu_0} \ {r\, \log^2( n)}}
\label{eta0}
\end{eqnarray}
observed entries.
\label{seed}
\end{lemma}
\begin{proof}
First note that from Theorem~\ref{mainthm}, the expected number of observed entries per column is at least
\begin{eqnarray*}
\eta = {\frac{128\, \beta \max\{\mu_1^2,\mu_0\}}{\nu_0} \ {r\, \log^2( n)}}
\end{eqnarray*}

Therefore, the number of observed entries $\widehat\eta$ in a column selected uniformly at random is probably not significantly less.  More precisely, by Chernoff's bound we have
$$\P(\widehat \eta \leq \eta/2) \ \leq \ \exp(-\eta/8) \ < \ e^{-4} \ . $$
Combining this with {\bf{A3}}, we have the probability that a randomly selected column belongs to
$\S_{i,\epsilon_0}$ and has $\eta/2$ or more observed entries is at least $\nu_0'/k$, where
$\nu_0':=(1-e^{-4})\nu_0$.
Then, the probability that the set of $s$ columns does not contain a column from $\S_{i,\epsilon_0}$ with at least $\eta/2$ observed entries is less than $(1-\nu_0'/k)^s$.  The probability that the set does not contain at least one column from $\S_{i,\epsilon_0}$ with $\eta/2$ or more observed entries, for $i=1,\dots,k$ is less than $\delta_0=k(1-\nu_0'/k)^s$ . Solving for $s$ in terms of $\delta_0$ yields
$$s \ = \ \frac{\log k + \log 1/\delta_0}{\log\left(\frac{k/\nu_0'}{k/\nu_0'-1}\right)}$$
The result follows by noting that $\log(x/(x-1)) \geq 1/x,$ for $x>1$.
\end{proof}

Next, for each seed we must find a set of $n$ columns from the same subspace as the seed.
This will be accomplished by identifying columns that are $\epsilon_0$-close to the seed, so that if the seed belongs to $\S_{i,\epsilon_0}$, the columns must belong to the same subspace.  Clearly the total number of columns $N$ must be sufficiently large so that $n$ or more such columns can be found.  We will return to the requirement on $N$ a bit later, after first dealing with the following challenge.

Since the columns are only partially observed, it may not be possible to determine how close each is to the seed.  We address this by showing that if a column and the seed are both observed on enough common indices, then the incoherence assumption {\bf{A2}} allows us reliably estimate the distance.
\begin{lemma}
Assume {\bf{A2}} and let $y=x_1-x_2$, where $x_1$ and $x_2$ are two columns of $\bX$.  Assume there is a common set of indices of size $q\leq n$ where both $x_1$ and $x_2 $ are observed.  Let $\omega$ denote this common set of indices and let $y_{\omega}$  denote the corresponding subset of $y$. Then for any $\delta_0>0$, if the number of commonly observed elements
$$q \ \geq \ 8 \mu_1^2 \log(2/\delta_0)  \ , $$
then with probability at least $1-\delta_0$
$$\frac{1}{2}\|y\|_2^2 \ \leq \ \frac{n}{q} \|y_{\omega}\|_2^2 \ \leq \ \frac{3}{2} \|y\|_2^2 \ .$$
\label{partial}
\end{lemma} \vspace{-.3in}
\begin{proof}
Note that $\|y_{\omega}\|_2^2$ is the sum of $q$ random variables drawn uniformly at random without replacement from the set $\{y_1^2,y_2^2,\dots,y_n^2\}$, and $\E \|y_{\omega}\|_2^2 = \frac{q}{n}\|y\|_2^2$.  We will prove the bound under the assumption that, instead, the $q$ variables are sampled with replacement, so that they are independent. By Lemma~\ref{meta}, this will provide the desired result. Note that if one variable in the sum $\|y_{\omega}\|_2^2$ is replaced with another value, then the sum changes in value by at most $2\|y\|_\infty^2$. Therefore, McDiramid's Inequality shows that  for $t>0$
$$\P\left(\left|\|y_\omega\|_2^2-\frac{q}{n}\|y\|_2^2\right|\geq t\right) \ \leq \ 2\,\exp\left(\frac{-t^2}{2q\|y\|_\infty^4}\right) \ , $$
or equivalently
$$\P\left(\left|\frac{n}{q}\|y_\omega\|_2^2-\|y\|_2^2\right|\geq t\right) \ \leq \ 2\,\exp\left(\frac{-qt^2}{2n^2\|y\|_\infty^4}\right) \ . $$
Assumption {\bf A2} implies that $n^2\|y\|_\infty^4 \leq \mu_1^2 \|y\|_2^4$, and so we have
$$\P\left(\left|\frac{n}{q}\|y_\omega\|_2^2-\|y\|_2^2\right|\geq t\right) \ \leq \ 2\,\exp\left(\frac{-qt^2}{2\mu_1^2\|y\|_2^4}\right) \ . $$
Taking $t=\frac{1}{2}\|y\|_2^2$ yields the result.
\end{proof}
Suppose that $x_1\in \S_{i,\epsilon_0}$ (for some $i$) and that $x_2 \not \in \S_i$, and
 that both $x_1,x_2$ observe $q\geq 2 \mu_0^2 \log(2/\delta_0)$ common indices. Let $y_\omega$ denote the difference between $x_1$ and $x_2$ on the common support set. If the {\em partial distance} $\frac{n}{q}\|y_\omega\|_2^2 \leq \epsilon_0^2/2$, then the result above implies that with probability at least $1-\delta_0$
$$\|x_1-x_2\|_2^2 \ \leq \ 2 \, \frac{n}{q} \, \|y_\omega\|_2^2 \ \leq \ \epsilon_0^2 . $$
On the other hand if $x_2 \in \S_i$ and $\|x_1-x_2\|_2^2 \leq \epsilon_0^2/3$, then with probability at least $1-\delta_0$
$$\frac{n}{q} \, \|y_\omega\|_2^2 \ \leq \ \frac{3}{2} \|x_1-x_2\|_2^2 \ \leq \ \epsilon_0^2/2 \ . $$
Using these results we will proceed as follows.  For each seed we find all columns that have at least \mbox{$t_0> 2 \mu_0^2 \log(2/\delta_0)$} observations at indices in common with the seed (the precise value of $t_0$ will be specified in a moment).  Assuming that this set is sufficiently large, we will select $\ell n$ these columns uniformly at random, for some integer $\ell \geq 1$.  In particular, $\ell$ will be chosen so that with high probability at least $n$ of the columns will be within $\epsilon_0/\sqrt{3}$ of the seed, ensuring that with probability at least $\delta_0$ the corresponding partial distance of each will be within $\epsilon_0/\sqrt{2}$.  That is enough to guarantee with the same probability that the columns are within $\epsilon_0$ of the seed.  Of course, a union bound will be needed so that the distance bounds above hold uniformly over the set of $s\ell n$ columns under consideration, which means that we will need each to have at least $t_0:=2\mu_0^2 \log(2s\ell n/\delta_0)$ observations at indices in common with the corresponding seed.   All this is predicated on $N$ being large enough so that such columns exist in $\bX_\Omega$.  We will return to this issue later, after determining the requirement for $\ell$.  For now we will simply assume that $N\geq \ell n$.

\begin{lemma}
Assume {\bf{A3}} and for each seed $x$ let $T_{x,\epsilon_0}$ denote the number of columns of $\bX$ in the ball of radius $\epsilon_0/\sqrt{3}$ about $x$.  If the number of columns selected for each seed, $\ell{n}$, such that,
\begin{eqnarray*}
\ell \ \geq \ \max\left\{\frac{2k}{\nu_0 (\frac{\epsilon_0}{\sqrt{3}})^r} \, , \, \frac{8k  \log(s/\delta_0)}{ n \nu_0 (\frac{\epsilon_0}{\sqrt{3}})^r} \right\} \ ,
\label{ell}
\end{eqnarray*}
then $\P\left(T_{x,\epsilon_0}  \leq  n \right) \ \leq \ \delta_0$ for all $s$ seeds.
\label{lemma:neighbors}
\end{lemma}
\begin{proof}
The probability that a column chosen uniformly at random from $\bX$ belongs to this ball is at least $\nu_0 (\epsilon_0/\sqrt{3})^r/k$, by Assumption {\bf{A3}}. Therefore the expected number of points is $$\E[T_{x,\epsilon_0}] \geq \frac{\ell n \nu_0 (\frac{\epsilon_0}{\sqrt{3}})^r}{k} \ .$$ By Chernoff's bound for any $0 < \gamma < 1$
\begin{eqnarray*}
\P\left(T_{x,\epsilon_0} \leq (1-\gamma)  \frac{\ell n \nu_0 (\frac{\epsilon_0}{\sqrt{3}})^r}{k} \right) \leq \exp\left(-\frac{\gamma^2}{2}   \frac{\ell n \nu_0 (\frac{\epsilon_0}{\sqrt{3}})^r}{k} \right) \ .
\end{eqnarray*}
Take $\gamma=1/2$ which yields
\begin{eqnarray*}
\P\left(T_{x,\epsilon_0} \leq  \frac{\ell n \nu_0 (\frac{\epsilon_0}{\sqrt{3}})^r}{2k} \right) & \leq & \exp\left(-  \frac{\ell n \nu_0 (\frac{\epsilon_0}{\sqrt{3}})^r}{8k} \right) \ .
\end{eqnarray*}
We would like to choose $\ell$ so that  $\frac{\ell n \nu_0 (\frac{\epsilon_0}{\sqrt{3}})^r}{2k} \geq n$ and
so that $\exp\left(-  \frac{\ell n \nu_0 (\frac{\epsilon_0}{\sqrt{3}})^r}{8k} \right) \leq \delta_0/s$ (so that the desired result fails for one or more of the $s$ seeds is less than $\delta_0$). The first condition leads to the requirement $\ell \geq \frac{2k}{\nu_0 (\frac{\epsilon_0}{\sqrt{3}})^r}$. The second condition produces the requirement
$\ell \ \geq \ \frac{8k  \log(s/\delta_0)}{n \nu_0 (\frac{\epsilon_0}{\sqrt{3}})^r}.$
\end{proof}

We can now formally state the procedure for finding local neighborhoods in Algorithm~\ref{alg:localNeighbor}.  Recall that the number of observed entries in each seed is at least $\eta_0$, per Lemma~\ref{seed}.

\begin{algorithm}[t]
\caption{- Local Neighborhood Procedure}\label{alg:localNeighbor} {\textbf {Input:}} $n$, $k$, $\mu_0$, $\epsilon_0$, $\nu_0$, $\eta_0$, $\delta_0 > 0$.
\begin{eqnarray*}
s_0 & := & \left\lceil \frac{k(\log k + \log 1/\delta_0)}{(1-e^{-4})\nu_0} \right\rceil \\
\ell_0 & := & \left\lceil \max\left\{\frac{2k}{\nu_0 (\frac{\epsilon_0}{\sqrt{3}})^r} \, , \, \frac{8k  \log(s_0/\delta_0)}{ n \nu_0 (\frac{\epsilon_0}{\sqrt{3}})^r} \right\} \right\rceil \\
t_0 & := & \lceil 2\mu_0^2 \log(2s_0\ell_0 n/\delta_0) \rceil
\end{eqnarray*}
{\bf Steps:}
\begin{enumerate}
\item Select $s_0$ ``seed'' columns uniformly at random and discard all with less than $\eta_0$
observations
\item For each seed, find all columns with $t_0$ observations at locations observed in the seed
\item Randomly select $\ell_0 n$ columns from each such set
\item Form local neighborhood for each seed by randomly \mbox{selecting} $n$ columns with partial distance less than $\epsilon_0/\sqrt{2}$ from the seed
\end{enumerate}
\label{alg1}
\end{algorithm}


\begin{lemma}
If $N$ is sufficiently large and $\eta_0>t_0$, then the Local Neighborhood Procedure in Algorithm~\ref{alg:localNeighbor} produces at least $n$ columns within $\epsilon_0$ of each seed, and at least one seed will belong
to each of $\S_{i,\epsilon_0}$, for $i=1,\dots,k$, with probability at least $1-3\delta_0$.
\label{NN}
\end{lemma}
\begin{proof}
Lemma~\ref{seed} states that if we select $s_0$ seeds, then with probability at least $1-\delta_0$ there is a seed in each $\S_{i,\epsilon_0}$, $i=1,\dots,k$, with at least $\eta_0$
observed entries, where $\eta_0$ is defined in (\ref{eta0}).
Lemma~\ref{lemma:neighbors} implies that if $\ell_0 n$ columns are selected uniformly at random for each seed, then with probability at least $1-\delta_0$ for each seed at least $n$ of the columns will be within a distance $\epsilon_0/\sqrt{3}$ of the seed.  Each seed has at least $\eta_0$ observed entries and we need to find $\ell_0 n$ other columns with at least $t_0$ observations at indices where the seed was observed.  Provided that $\eta_0\geq t_0$, this is certainly possible if $N$ is large enough. It follows from Lemma~\ref{partial} that $\ell_0 n$ columns have at least $t_0$ observations at indices where the seed was also observed, then with probability at least $1-\delta_0$ the partial distances will be within $\epsilon_0/\sqrt{2}$, which implies the true distances are within $\epsilon_0$.    The result follows by the union bound.
\end{proof}
Finally, we quantify just how large $N$ needs to be. Lemma~\ref{lemma:neighbors} also shows that we require at least
$$N \ \geq \ \ell n \ \geq \  \max\left\{\frac{2kn}{\nu_0 (\frac{\epsilon_0}{\sqrt{3}})^r} \, , \, \frac{8k  \log(s/\delta_0)}{\nu_0 (\frac{\epsilon_0}{\sqrt{3}})^r} \right\} \ . $$
However, we must also determine a lower bound on the probability that a column selected uniformly at random has at least $t_0$ observed indices in common with a seed.
Let $\gamma_0$ denote this probability, and let $p_0$ denote the probability of observing each entry in $\bX$. Note that our main result, Theorem~\ref{mainthm}, assumes that
\begin{eqnarray*}
p_0 \geq  \ \p0 \ .
\end{eqnarray*}
Since each seed has at least $\eta_0$ entries observed, $\gamma_0$ is greater than or equal to the probability that a $\mbox{Binomial}(\eta_0,p_0)$ random variable is at least $t_0$. Thus,
$$\gamma_0 \ \geq \ \sum_{j= t_0}^{\eta_0} {\eta_0\choose{j}}p_0^{j}(1-p_0)^{\eta_0-j} \ . $$
This implies that the expected number of columns with $t_0$ or more observed indices in common with a seed is at least $\gamma_0 N$.   If $\widetilde n$ is the actual number with this property,
then by Chernoff's bound, $\P(\widetilde n \leq \gamma_0 N/2) \leq \exp(-\gamma_0 N/8)$.
So $N\geq 2 \ell_0 \gamma_0^{-1} n$ will suffice to guarantee that enough columns can be found
for each seed with probability at least $1-s_0 \exp(-\ell_0 n/4) \geq 1-\delta_0$ since this will be far larger than $1-\delta_0$, since $\delta_0$ is polynomial in $n$.

To take this a step further, a simple lower bound on $\gamma_0$ is obtained as follows.  Suppose we consider only a $t_0$-sized subset of the indices where the seed is observed.  The probability that another column selected at random is observed at all $t_0$ indices in this subset is $p_0^{t_0}$.  Clearly $\gamma_0 \geq p_0^{t_0} = \exp(t_0 \log p_0) \geq (2s_0\ell_0 n)^{2 \mu_0^2 \log p_0}$. This yields the following sufficient condition on the size of $N$:
$$N \ \geq \  \ell_0 n (2s_0\ell_0 n/\delta_0)^{2 \mu_0^2 \log p_0^{-1}} \ . $$

From the definitions of $s_0$ and $\ell_0$, this implies that if $2\mu_0^2 \log p_0$ is a fixed constant, then a sufficient number of columns will exist if $N=O(\mbox{poly}(kn/\delta_0))$.  For example, if $\mu_0^2 = 1$ and $p_0 =1/2$, then $N=O((kn)/\delta_0)^{2.4})$ will suffice; i.e., $N$ need only grow polynomially in $n$.  On the other hand, in the extremely undersampled case $p_0$ scales like $\log^2(n)/n$ (as $n$ grows and $r$ and $k$ stay constant) and $N$ will need to grow almost exponentially in $n$, like $n^{\log n - 2\log\log n}$.

\section{Local Subspace Completion}
\label{sec:localsubspaces}

For each of our local neighbor sets, we will have an incompletely observed $n \times n$ matrix; if all the neighbors belong to a single subspace, the matrix will have rank $\leq r$.
First, we recall the following result from low-rank matrix completion theory \cite{mcRecht}.
\begin{lemma} \label{thm:mc} Consider an $n \times n$ matrix of rank $\leq r$ and row and column spaces with coherences bounded above by some constant $\mu_0$.  Then the matrix can be exactly completed if
\begin{eqnarray}
m' \geq 64\max{\left(\mu_1^2,\mu_0\right)} \beta r n\log^2\left(2n\right)
\end{eqnarray}
 entries are observed uniformly at random, for constants $\beta>0$ and with probability
$\geq 1-6 \left(2n\right)^{2-2\beta}\log{n}-n^{2-2\beta^{1/2}}$.
\end{lemma}

We wish to apply these results to our local neighbor sets, but we have three issues we must address: First, the sampling of the matrices formed by local neighborhood sets is not uniform since the set is selected based on the observed indices of the seed. Second, given Lemma~\ref{seed} we must complete not one, but $s_0 $ (see Algorithm~\ref{alg:localNeighbor}) incomplete matrices simultaneously with high probability. Third, some of the local neighbor sets may have columns from more than one subspace. Let us consider each issue separately.

First consider the fact that our incomplete submatrices are not sampled uniformly.  The non-uniformity can be corrected with a simple thinning procedure. Recall that the columns in the seed's local neighborhood are identified first by finding columns with sufficient overlap with each seed's observations. To refer to the seed's observations, we will say ``the support of the seed.''

Due to this selection of columns, the resulting neighborhood columns are highly sampled on the support of the seed. In fact, if we again use the notation $q$ for the minimum overlap between two columns needed to calculate distance, then these columns have at least $q$ observations on the support of the seed. Off the support, these columns are still sampled uniformly at random with the same probability as the entire matrix. Therefore we focus only on correcting the sampling pattern on the support of the seed.

Let $t$ be the cardinality of the support of a particular seed. Because all entries of the entire matrix are sampled independently with probability $p_0$, then for a randomly selected column, 
the random variable which generates $t$ is binomial. For neighbors selected to have at least $q$ overlap with a particular seed, we denote $t'$ as the number of samples overlapping with the support of the seed. The probability density for $t'$ is positive only for $j=q, \dots, t$,
$$\P(t'=j)  = \frac{{t\choose{j}} p_0^{j}(1-p_0)^{t-j} }{\rho}$$
%
where $\rho = \sum_{j=q}^t{t\choose{j}} p_0^{j}(1-p_0)^{t-j} $.

In order to thin the common support, we need two new random variables. The first is a bernoulli, call it $Y$, which takes the value 1 with probability $\rho$ and 0 with probability $1-\rho$. 
The second random variable, call it $Z$, takes values $j=0,\dots, q-1$ with probability
\begin{eqnarray*}
\P(Z=j) = \frac{{t\choose{j}}p_0^{j}(1-p_0)^{t-j}}{1-\rho}
\end{eqnarray*}

Define $t'' = t'  Y + Z  (1-Y)$. The density of $t''$ is
\begin{equation}
  \P(t''=j) = \left\{ \begin{array}{cl}
	\P(Z=j) (1-\rho) &  j = 0, \dots, q-1 \\
	\P(t'=j) \rho &  j = q, \dots, t \\
\end{array} \right.
\end{equation}
which equal to the desired binomial distribution.  Thus, the thinning is accomplished as follows.
For each column draw an independent sample of $Y$.  If the sample is $1$, then the column is not altered.  If the sample is $0$, then a realization of $Z$ is drawn, which we denote by $z$.  Select a random subset of size $z$ from the observed entries in the seed support and discard the remainder.  We note that the seed itself should not be used in completion, because there is a dependence between the sample locations of the seed column and its selected neighbors which cannot be eliminated.


Now after thinning, we have the following matrix completion guarantee for each neighborhood matrix.
\begin{lemma}
\label{lemma:mcNN}
Assume all $s_0$ seed neighborhood matrices are thinned according to the discussion above, have rank $\leq r$, and the matrix entries are observed uniformly at random with probability,
\begin{eqnarray}
p_0 \geq {\frac{128\, \beta \max\{\mu_1^2,\mu_0\}}{\nu_0} \ \frac{r\, \log^2( n)}{n}}
\end{eqnarray}
Then with probability
$\geq 1- 12 s_0 n^{2-2\beta^{1/2}}\log{n}$, all $s_0$ matrices can be perfectly completed.
\end{lemma}
\begin{proof}
First, we find that if each matrix has
\begin{eqnarray*}
m' \geq 64\max{\left(\mu_1^2,\mu_0\right)} \beta rn \log^2\left(2n\right)
\end{eqnarray*}
entries observed uniformly at random (with replacement), then with probability $\geq 1- 12 s_0 n^{2-2\beta^{1/2}}\log{n}$, all $s_0$ matrices are perfectly completed.  This follows by Lemma~\ref{thm:mc}, the observation that $$6 \left(2n\right)^{2-2\beta}\log{n}+n^{2-2\beta^{1/2}} \leq 12 n^{2-2\beta^{1/2}}\log{n} \ , $$ and a simple application of the union bound.

But, under our sampling assumptions, the number of entries observed in each seed neighborhood matrix is random.  Thus, the total number of observed entries in each is guaranteed to be sufficiently large with high probability as follows.  The random number of entries observed in an $n\times n$ matrix is $\widehat{m}\sim \mbox{Binomial}(p_0,n^2)$.  By Chernoff's bound we have
$\P(\widehat{m}\leq n^2 p_0/2) \leq \exp(-n^2 p_0/8)$.   By the union bound we find that $\widehat{m} \geq m'$ entries are observed in each of the $s_0$ seed matrices with probability at least $1-\exp(-n^2p_0/8+\log s_0)$ if $p_0 \geq {\frac{128\, \beta \max\{\mu_1^2,\mu_0\}}{\nu_0} \ \frac{r\, \log^2( n)}{n}}$.

Since $n^2p_0>rn\log^2 n$ and $s_0=O(k(\log k+\log n))$, this probability tends to zero exponentially in $n$ as long as $k=o(e^n)$, which holds according to  Assumption {\bf A1}.  Therefore this holds with probability at least $1-12 s_0 n^{2-2\beta^{1/2}}\log{n}$.
\end{proof}

Finally, let us consider the third issue, the possibility that one or more of the points in the neighborhood of a seed lies in a subspace different than the seed subspace.  When this occurs, the rank of the submatrix formed by the seed's neighbor columns will be larger than the dimension of the seed subspace. Without loss of generality assume that we have only two subspaces represented in the neighbor set, and assume their dimensions are $r'$ and $r''$.  First, in the case that $r' + r'' > r$, when a rank $\geq r$ matrix is completed to a rank $r$ matrix, with overwhelming probability there will be errors with respect to the observations as long as the number of samples in each column is $O(r \log r)$, which is assumed in our case; see~\cite{balzanoSubspace}. Thus we can detect and discard these candidates.  Secondly, in the case that $r' + r'' \leq r$, we still have enough samples to complete this matrix successfully with high probability. However, since we have drawn enough seeds to guarantee that every subspace has a seed with a neighborhood entirely in that subspace, we will find that this problem seed is redundant.  This is determined in the Subspace Refinement step.

\section{Subspace Refinement}
\label{sec:subspacerefine}

Each of the matrix completion steps above yields a low-rank matrix with a corresponding column subspace, which we will call the {\em candidate} subspaces.
While the true number of subspaces will not be known in advance, since $s_0=O(k(\log k+\log(1/\delta_0))$, the candidate subspaces will contain the true subspaces with high probability (see Lemma~\ref{lemma:neighbors}). We must now deal with the algorithmic issue of determining the true set of subspaces.

We first note that, from Assumption {\bf{A3}}, with probability 1 a set
of points of size $\geq r$ all drawn from a single subspace ${\cal S}$ of dimension $\leq r$ will span ${\cal S}$. In fact, any $b<r$ points will span a $b$-dimensional subspace of the $r$-dimensional subspace ${\cal S}$.

Assume that $r<n$, since otherwise it is clearly necessary to observe all entries. Therefore, if a seed's nearest neighborhood set is confined to a single subspace, then the columns in span their subspace. And if the seed's nearest neighborhood contains columns from two or more subspaces, then the matrix will have rank larger than that of any of the constituent subspaces.
Thus, if a certain candidate subspace is spanned by the union of two or more smaller candidate subspaces, then it follows that that subspace is not a true subspace (since we assume that none of the true subspaces are contained within another).

This observation suggests the following subspace refinement procedure. The $s_0$ matrix completions yield $s\leq s_0$ candidate column subspaces; $s$ may be less than $s_0$ since completions that fail are discarded as described above.
First sort the estimated subspaces in order of rank from smallest to largest (with arbitrary ordering of subspaces of the same rank), which we write as $\S_{(1)},\dots,\S_{(s)}$. We will denote the final set of estimated subspaces as $\widehat{\S}_1,\dots,\widehat{\S}_k$.   The first subspace $\widehat{\S}_1 := \S_{(1)}$, a lowest-rank subspace in the candidate set.  Next, $\widehat{\S}_2 = \S_{(2)}$ if and only if $\S_{(2)}$ is not contained in $\widehat{\S}_1$.   Following this simple sequential strategy, suppose that when we reach the candidate  $\S_{(j)}$ we have so far determined $\widehat{\S}_1,\dots,\widehat{\S}_i$, $i<j$.  If $\S_{(j)}$ is not in the span of $\cup_{\ell=1}^i \widehat{\S}_\ell$, then we set $\widehat{\S}_{i+1} = \S_{(j)}$, otherwise we move on to the next candidate.
In this way, we can proceed sequentially through the rank-ordered list of candidates, and we will identify all true subspaces.

\section{The Full Monty}
\label{sec:fullmonty}

Now all will be revealed.  At this point, we have identified the true subspaces, and all $N$ columns lie in the span of one of those subspaces.  For ease of presentation, we assume that the number of subspaces is exactly $k$. However if columns lie in the span of fewer than $k$, then the procedure above will produce the correct number. To complete the full matrix, we proceed one column at a time. For each column of $\bX_{\Omega}$, we determine the correct subspace to which this column belongs, and we then complete the column using that subspace. We can do this with high probability due to results from~\cite{balzanoSubspace,balzanoszlam11}.

The first step is that of subspace assignment, determining the correct subspace to which this column belongs. In~\cite{balzanoszlam11}, it is shown that given $k$ subspaces, an incomplete vector can be assigned to its closest subspace with high probability given enough observations. In the situation at hand, we have a special case of the results of~\cite{balzanoszlam11} because we are considering the more specific situation where our incomplete vector lies exactly in one of the candidate subspaces, and we have an upper bound for both the dimension and coherence of those subspaces.

\begin{lemma}
Let $\{{\cal S}_1, \dots, {\cal S}_{
k}\}$ be a collection of $k$ subspaces of dimension $\leq r$ and coherence parameter bounded above by $\mu_0$.  Consider column vector $x$ with index set $\Omega \in \{1, \dots, n\}$, and define $P_{\Omega,{\cal S}_j} = U^{j}_\Omega \left( \left(U^{j}_\Omega\right)^T U^{j}_\Omega \right)^{-1} \left(U^{j}_\Omega\right)^T$, where $U^{j}$ is the orthonormal column span of ${\cal S}_j$ and $U^{j}_\Omega$ is the column span of ${\cal S}_j$ restricted to the observed rows, $\Omega$.  Without loss of generality, suppose the column of interest $x \in {\cal S}_1$.  If {\bf{A3}} holds, and the probability of observing each entry of $x$ is independent and Bernoulli with parameter $$p_0 \geq \p0\;.$$  Then with probability at least $1-(3(k-1)+2)\delta_0$,
\begin{equation}
\|x_\Omega - P_{\Omega,{\cal S}_1} x_\Omega\|_2^2 = 0
\label{eq:incompleteResid}
\end{equation}
and for $j=2,\dots,k$
\begin{equation}
\|x_\Omega - P_{\Omega,{\cal S}_j} x_\Omega\|_2^2 \ > \ 0 \ .
\label{eq:assignment}
\end{equation}
\label{thm:projection}
\end{lemma}
\begin{proof}
We wish to use results from~\cite{balzanoSubspace,balzanoszlam11}, which require a fixed number of measurements $|\Omega|$. By Chernoff's bound $$\P\left(|\Omega| \leq \frac{np_0}{2}\right) \ \leq \ \exp\left(\frac{-np_0}{8}\right).$$
Note that $np_0 > 16r\beta\log^2n$, therefore $\exp\left(\frac{-np_0}{8}\right) < (n^{-2\beta})^{\log n} < \delta_0$;
in other words, we observe $|\Omega|>np_0/2$ entries of $x$ with probability $1-\delta_0$. This set $\Omega$ is selected uniformly at random among all sets of size $|\Omega|$, but using Lemma~\ref{meta} we can assume that the samples are drawn uniformly with replacement in order to apply results of~\cite{balzanoSubspace,balzanoszlam11}.

Now we show that $|\Omega|>np_0/2$ samples selected uniformly with replacement implies that
\begin{equation}
|\Omega|  > \max\left\{\frac{8 r\mu_0}{3}  \log\left(\frac{2r}{\delta_0}\right), \frac{r\mu_0 (1+\xi)^2}{(1-\alpha)(1-\gamma)}\right\}
\label{eq:projlemreq}
\end{equation}
where $\xi, \alpha > 0$ and $\gamma \in (0,1)$ are defined as $\alpha = \sqrt{\frac{2 \mu_1^2}{|\Omega|} \log\left(\frac{1}{\delta_0}\right)}$, $\xi = \sqrt{2\mu_1\log\left(\frac{1}{\delta_0}\right)}$, and $\gamma = \sqrt{\frac{8r\mu_0}{3|\Omega|} \log\left(\frac{2r}{\delta_0}\right)}$.

We start with the second term in the max of (\ref{eq:projlemreq}). Substituting $\delta_0$ and the bound for $p_0$, one can show that for $n \geq 15$ both $\alpha \leq 1/2$ and $\gamma \leq 1/2$. This makes $(1+\xi)^2 / (1-\alpha)(1-\gamma) \leq 4(1-\xi)^2 \leq 8 \xi^2$ for $\xi>2.5$, i.e. for $\delta_0< 0.04$.

We finish this argument by noting that $8 \xi^2 = 16 \mu_1 \log(1/\delta_0) < n p_0 / 2$; there is in fact an $O(r\log(n))$ gap between the two. Similarly for the first term in the max of (\ref{eq:projlemreq}), $\frac{8}{3} r \mu_0  \log\left(\frac{2r}{\delta_0}\right) < n p_0 / 2$; here the gap is $O(\log(n))$.

Now we prove (\ref{eq:incompleteResid}), which follows from~\cite{balzanoSubspace}. With  $|\Omega|  > \frac{8}{3} r \mu_0 \log\left(\frac{2r}{\delta_0}\right)$, we have that $ U_\Omega^T U_\Omega$ is invertible with probability at least $1-\delta_0$ according to Lemma~3 of~\cite{balzanoSubspace}. This implies that
\begin{equation}
U^T x = \left( U_\Omega^T U_\Omega\right)^{-1} U_\Omega^T x_\Omega \;.
\label{eq:sameweights}
\end{equation}
Call $a_1 = U^Tx$. Since $x \in {\cal S}$, $Ua_1 = x$, and $a_1$ is in fact the unique solution to $Ua = x$.
Now consider the equation $U_\Omega a = x_\Omega$. The assumption that $U_\Omega^T U_\Omega$ is invertible implies that $a_2 = \left( U_\Omega^T U_\Omega\right)^{-1} U_\Omega^T x_\Omega$ exists and is the unique solution to $U_\Omega a = x_\Omega$.
However, $U_\Omega a_1 = x_\Omega$ as well, meaning that $a_1=a_2$. Thus, we have
$$\|x_\Omega - P_{\Omega,{\cal S}_1} x_\Omega\|_2^2 = \|x_\Omega - U_\Omega U^T x\|_2^2 = 0$$
with probability at least $1-\delta_0$.

Now we prove (\ref{eq:assignment}), paralleling Theorem~1 in~\cite{balzanoszlam11}. We use Assumption {\bf{A3}} to ensure that $x \notin {\cal S}_j$, $j=2,\dots,k$. This along with (\ref{eq:projlemreq}) and Theorem~1 from~\cite{balzanoSubspace} guarantees that
\begin{eqnarray*}
 \|x_\Omega - P_{\Omega,{\cal S}_j} x_\Omega\|_2^2 \hfill \geq \frac{ |\Omega| (1-\alpha) - r \mu_0 \frac{(1+\xi)^2}{ 1-\gamma}}{n}  \|x - P_{{\cal S}_j} x\|_2^2> 0
\end{eqnarray*}
for each $j = 2, \dots, k$ with probability at least $1-3\delta_0$. With a union bound this holds simultaneously for all $k-1$ alternative subspaces with probability at least $1-3(k-1)\delta_0$. When we also include the events that (\ref{eq:incompleteResid}) holds and that $|\Omega|>np_0/2$, we get that the entire theorem holds with probability at least $1-(3(k-1)+2)\delta_0$. \end{proof}

Finally, denote the column to be completed by $x_\Omega$. To complete $x_\Omega$ we first
determine which subspace it belongs to using the results above.  For a given column we can use the \emph{incomplete data projection residual} of (\ref{eq:incompleteResid}).  With probability at least $1-(3(k-1)+2)\delta_0$, the residual will be zero for the correct subspace and strictly positive for all other subspaces.  Using the span of the chosen subspace, $U$, we can then complete the column by using $\widehat{x} = U \left( U_\Omega^T U_\Omega\right)^{-1} U_\Omega^T x_\Omega$.

We reiterate that Lemma~\ref{thm:projection} allows us to complete a single column $x$ with probability $1-(3(k-1)+2)\delta_0$. If we wish to complete the entire matrix, we will need another union bound over all $N$ columns, leading to a $\log N$ factor in our requirement on $p_0$.
Since $N$ may be quite large in applications, we prefer to state our result in terms of per-column completion bound.

The confidence level stated in Theorem~\ref{mainthm} is the result of applying the union bound to all the steps required in the Sections~3, 4, and 6.  All hold simultaneously with probability at least
\begin{eqnarray*}
1- (6+3(k-1)+12s_0)\, \delta_0 & < & 1-(6+15s_0)\delta_0 \ ,
\end{eqnarray*}
which proves the theorem.

\section{Experiments}
\label{sec:exp}

The following experiments evaluate the performance of the proposed high-rank matrix completion procedure and compare results with standard low-rank matrix completion based on nuclear norm minimization.

\subsection{Numerical Simulations}

We begin by examining a highly synthesized experiment where the data exactly matches the assumptions of our high-rank matrix completion procedure.  The key parameters were chosen as follows: $n = 100$, $N=5000$, $k=10$, and $r=5$.
The $k$ subspaces were $r$-dimensional, and each was generated by $r$ vectors drawn from the
${\cal N}(0,I_{n})$ distribution and taking their span.   The resulting subspaces are highly incoherent with the canonical basis for $\R^{n}$.
For each subspace, we generate $500$ points drawn from a ${\cal N}(0,UU^{T})$ distribution,
where $U$ is a $n\times r$ matrix whose orthonormal columns span the subspace.  Our procedure was implemented using $\lceil 3k\log k\rceil$ seeds.  The matrix completion software called GROUSE (available here \cite{grouse}) was used in our procedure and to implement the standard low-rank matrix completions.  We ran $50$ independent trials of our procedure and compared it to standard low-rank matrix completion.  The results are summarized in the figures below.  The key message is that our new procedure can provide accurate completions from far fewer observations compared to standard low-rank completion, which is precisely what our main result predicts.

\begin{figure}[htb]
\begin{minipage}[h]{0.99\linewidth}
\centerline{\includegraphics[width=6cm]{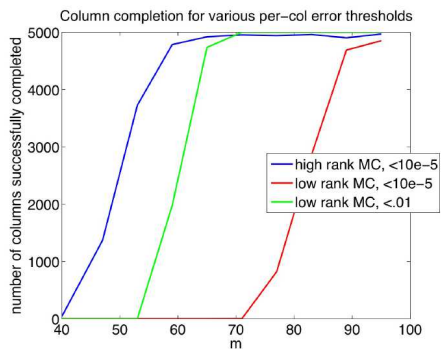}}
\end{minipage}
\caption{\label{fig:synthetic} \small The number of correctly completed columns (with tolerances shown above, $10$e-$5$ or $0.01$), versus the average number of observations per column.  As expected, our procedure (termed high rank MC in the plot) provides accurate completion with only about $50$ samples per column.  Note that $r\log n \approx 23$ in this simulation, so this is quite close to our bound.  On the other hand, since the rank of the full matrix is $rk=50$, the standard low-rank matrix completion bound requires $m> 50\log n \approx 230$.  Therefore, it is not surprising that the standard method (termed low rank MC above) requires almost all samples in each column.  }
\end{figure}


\subsection{Network Topology Inference Experiments}

The ability to recover Internet router-level connectivity is of importance to network managers, network operators and the area of security.  As a complement to the heavy network load of standard active probing methods ({\em e.g.,} \cite{rocket}), which scale poorly for Internet-scale networks, recent research has focused on the ability to recover Internet connectivity from passively observed measurements \cite{EBNC07}.  Using a passive scheme, no additional probes are sent through the network; instead we place passive monitors on network links to observe ``hop-counts'' in the Internet ({\em i.e,} the number of routers between two Internet resources) from traffic that naturally traverses the link the monitor is placed on.  An example of this measurement infrastructure can be seen in Figure~\ref{fig:networkExample}.

\begin{figure}[htb]
\begin{minipage}[h]{0.99\linewidth}
\centerline{\includegraphics[width=6cm]{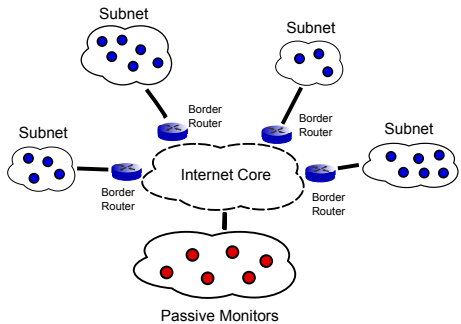}}
\end{minipage}
\caption{\label{fig:networkExample} \small Internet topology example of subnets sending traffic to passive monitors through the Internet core and common border routers.}
\end{figure}

These hop count observations result in an $n \times N$ matrix, where $n$ is the number of passive monitors and $N$ is the total unique IP addresses observed.  Due to the passive nature of these observations, specifically the requirement that we only observe traffic that happens to be traversing the link where a monitor is located, this hop count matrix will be massively incomplete.  A common goal is to impute (or {\em fill-in}) the missing components of this hop count matrix in order to infer network characteristics.

Prior work on analyzing passively observed hop matrices have found a distinct subspace mixture structure \cite{EBN08}, where the full hop count matrix, while globally high rank, is generated from a series of low rank subcomponents.  These low rank subcomponents are the result of the Internet topology structure, where all IP addresses in a common subnet exist behind a single common border router.  This network structure is such that any probe sent from an IP in a a particular subnet to a monitor must traverse through the same border router.  A result of this structure is a rank-two hop count matrix for all IP addresses in that subnet, consisting of the hop count vector to the border router and a constant offset relating to the distance from each IP address to the border router.  Using this insight, we apply the high-rank matrix completion approach on incomplete hop count matrices.

Using a Heuristically Optimal Topology from \cite{liSigcomm}, we simulated a network topology and measurement infrastructure consisting of $N=2700$ total IP addresses uniformly distributed over $k=12$ different subnets.  The hop counts are generated on the topology using shortest-path routing from $n=75$ passive monitors located randomly throughout the network.  As stated above, each subnet corresponds to a subspace of dimension $r=2$.   Observing only 40\% of the total hop counts, in Figure~\ref{fig:networkOrbis} we present the results of the hop count matrix completion experiments, comparing the performance of the high-rank procedure with standard low-rank matrix completion.  The experiment shows dramatic improvements, as over 70\% of the missing hop counts can be imputed exactly using the high-rank matrix completion methodology, and approximately no missing elements are imputed exactly using standard low-rank matrix completion.

\begin{figure}[htb!]
\begin{minipage}[h]{0.99\linewidth}
\centerline{\includegraphics[width=6cm]{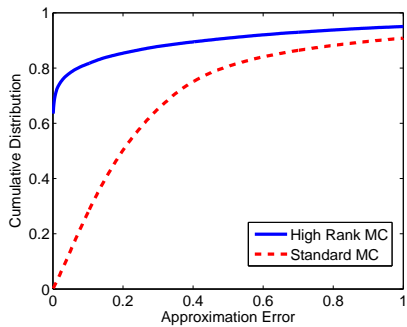}}
\end{minipage}
\caption{\label{fig:networkOrbis} \small Hop count imputation results, using a synthetic network with $k=12$ subnets, $n=75$ passive monitors, and $N=2700$ IP addresses.  The cumulative distribution of estimation error is shown with respect to observing 40\% of the total elements.}
\end{figure}

Finally, using real-world Internet delay measurements (courtesy of \cite{ledlieWild}) from $n=100$ monitors to $N=22550$ IP addresses, we test imputation performance when the underlying subnet structure is not known.  Using the estimate $k=15$, in Figure~\ref{fig:networkLatency} we find a significant performance increase using the high-rank matrix completion technique.

\begin{figure}[htb!]
\begin{minipage}[h]{0.99\linewidth}
\centerline{\includegraphics[width=6cm]{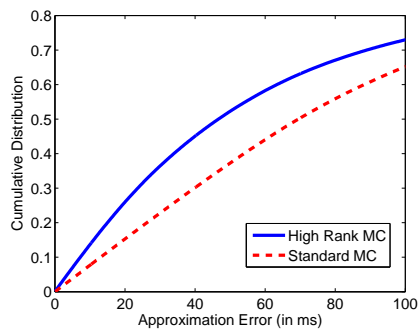}}
\end{minipage}
\caption{\label{fig:networkLatency} \small Real-world delay imputation results, using a network  $n=100$ monitors, $N=22550$ IP addresses, and an unknown number of subnets.  The cumulative distribution of estimation error is shown with respect to observing 40\% of the total delay elements.}
\end{figure}

%

\bibliographystyle{IEEEtran}
\bibliography{highrankMC12_15}

\end{document}